\documentclass[twoside, 11pt]{article}
\usepackage{mathrsfs}
\usepackage{amssymb, amsmath, mathrsfs, amsthm}%, wasysym}
\usepackage{graphicx}
\usepackage{amsmath}
\usepackage{color}
\usepackage[top=2cm, bottom=2cm, left=2cm, right=2cm]{geometry}
\usepackage{float, caption}%, subcaption
\usepackage{listings}
\usepackage{graphicx}
\usepackage[ruled,linesnumbered]{algorithm2e}
\input{epsf}

\newcommand{\bd}{\begin{description}}
\newcommand{\ed}{\end{description}}
\newcommand{\bi}{\begin{itemize}}
\newcommand{\ei}{\end{itemize}}
\newcommand{\be}{\begin{enumerate}}
\newcommand{\ee}{\end{enumerate}}
\newcommand{\beq}{\begin{equation}}
\newcommand{\eeq}{\end{equation}}
\newcommand{\beqs}{\begin{eqnarray*}}
\newcommand{\eeqs}{\end{eqnarray*}}

\numberwithin{equation}{section}

\definecolor{DarkGreen}{rgb}{0.2, 0.6, 0.3}

\newtheorem{theorem}{Theorem}[section]

\newtheorem{definition}{Definition}[section]
\newtheorem{corollary}[theorem]{Corollary}
\newtheorem{case}{Case}

\newtheorem{claim}{Claim}
\newtheorem{remark}{Remark}[section]
\newtheorem{example}{Example}[section]

\newtheorem{proposition}{Proposition}[section]

\newtheorem{observation}{Observation}[section]
\usepackage{tikz}
\usetikzlibrary{matrix}
\usepackage{pgffor}%可以使用foreach的包
\usepackage{ifthen}%可以使用ifthenelse的包，还能使用whiledo
\usetikzlibrary{positioning, shapes.geometric}
\usetikzlibrary{graphs, positioning, shapes.geometric}
\tikzstyle{sted} = [rectangle , rounded corners,
minimum size=1cm, align=center,text centered, draw = black, fill =yellow!20,scale=0.8]
\tikzstyle{rec} = [rectangle, text centered, minimum size=1.8cm, align=center, draw = black, fill = pink!20,scale=0.8]
\tikzstyle{blu} = [rectangle, text centered, minimum size=1.2cm, align=center, draw = black, fill = blue!10,scale=0.8]
\tikzstyle{gre} = [rectangle, text centered, minimum size=1.2cm, align=center, draw = black, fill = green!10,scale=0.8]

\long\def\delete#1{}

\input amssym.def
\input amssym.tex

\def\b{\beta}

\def\b0{\mathbf{0}}

\usepackage{xcolor}
\usepackage{booktabs}
\usepackage{multirow}
\usepackage{tabularx}
\usepackage[normalem]{ulem}
\usetikzlibrary{positioning}

\usepackage{changepage} % 加载 changepage 宏包
\usepackage{lipsum}
\usepackage{parskip} % 加载 parskip 宏包
\usepackage{setspace} % 加载 setspace 宏包
\usepackage{graphicx}
\usepackage{float}
\usepackage{placeins} % 加载 placeins 宏包以使用 \FloatBarrier
\usepackage{caption} % 加载 caption 宏包以使用 \captionof
\begin{document}
\title{\textbf{The $g$-good-neighbor  diagnosability of product networks under the PMC model} \footnote{Supported by the National Science Foundation of China
(Nos. 12471329 and 12061059).} }

\author{Zhao
Wang\thanks	{College of Science, China Jiliang University, Hangzhou
310018, China. {\tt
wangzhao@mail.bnu.edu.cn}}, \ \ Yaping Mao \footnote{Corresponding author: Academy of Plateau
Science and Sustainability, and School of Mathematics and Statistis, Qinghai
Normal University, Xining, Qinghai 810008, China. {\tt maoyaping@ymail.com}}, \ \ Sun-Yuan Hsieh\footnote{Department of Computer Science and Information Engineering, National Cheng Kung University, Tainan 701, Taiwan. {\tt hsiehsy@mail.ncku.edu.tw}},
\ \ Ralf Klasing\footnote{Universit\'{e} de Bordeaux, Bordeaux INP, CNRS, LaBRI, UMR 5800, Talence, France.
{\tt Email: ralf.klasing@labri.fr.}}
}
\date{}
\maketitle
\begin{abstract}
The concept of neighbor connectivity originated from the assessment of the subversion of espionage networks caused by underground resistance movements, and it has now been applied to measure the disruption of networks caused by cascading failures through neighbors. In this paper, we give two necessary and sufficient conditions of the existance of $g$-good-neighbor diagnosability. We introduce a new concept called $g$-good neighbor cut-component number (gc number for short), which has close relation with $g$-good-neighbor diagnosability. Sharp lower and upper bounds of the gc number of general graphs in terms of the $g$-good neighbor connectivity is given, which provides a formula to compute the $g$-good-neighbor diagnosability for general graphs (therefore for Cartesian product graphs). As their applications, we get the exact values or bounds for the gc numbers and $g$-good-neighbor diagnosability of grid, torus networks and generalized cubes.\\[0.1cm]
\end{abstract}

\section{Introduction}

Advances in semiconductor process technology have led to the development of multiprocessor systems.
Therefore, fault-diagnosis of multiprocessor systems becomes increasingly important.
A multiprocessor system consists of several processors, also known as nodes. These nodes communicate with each other by sending and receiving messages over the Internet. Within the Internet framework, each node is equipped with one or more processing elements, a local memory, and a communication module.

An undirected graph $G=G(V, E)$ serves as a model for a multiprocessor system. In this graph, processors are represented by vertices, and communication links are depicted by edges.

In the classic PMC model \cite{PMC67}, processors that are neighbors can test one another. For any two adjacent vertices $u$ and $v$ in the vertex set $V$, the ordered pair $(u, v)$ indicates that the processor associated with vertex $u$ is testing the processor represented by vertex $v$. Here, $u$ is called the \textit{tester}, and $v$ is known as the \textit{tested vertex}. Concerning the result of the test $(u, v)$, it is set to 1 if $u$ judges $v$ to be faulty; otherwise, it is set to 0 if $u$ deems $v$ to be fault-free.

For a system modeled by the graph $G=G(V, E)$, a \textit{test assignment} is composed of a set of tests $(u, v)$ for specific pairs of adjacent vertices. This test assignment can be presented as a directed graph $T=(V, L)$. In this directed graph, if $(u, v) \in L$, it implies that in the original graph $G$, vertices $u$ and $v$ are adjacent. We assume that for every edge connecting two vertices in $G$, the two vertices test each other, and all these tests are part of the test assignment.

The collection of all test outcomes for a test assignment $T$ is known as a \textit{syndrome}. In a formal sense, a syndrome is defined as a function $\sigma$ that maps from the set $L$ to the set $\{0,1\}$. The set consisting of all the processors that are malfunctioning within the system is referred to as the \textit{faulty set}. This faulty set can be any subset of the set $V$. The procedure of determining all the faulty vertices is what we call the diagnosis of the system. The largest number of faulty vertices that the system $G$ can ensure to be identified is called the \textit{diagnosability} of $G$,  is denoted as $t(G)$.

Given a syndrome $\sigma$, a vertex subset $F \subseteq V$ is considered \textit{consistent} with $\sigma$ under the following condition: for any link $(u, v) \in L$ where $u \in V-F, \sigma(u, v)=1$ exactly when $v \in F$. Since a malfunctioning tester can yield unreliable results, a particular set $F$ of faulty vertices might generate various syndromes. Denote by $\sigma(F)$ the set of all syndromes that could be produced when $F$ is the set of faulty vertices.

Two different vertex sets $F_1, F_2 \subset V$ are called \textit{indistinguishable} if the intersection of $\sigma\left(F_1\right)$ and $\sigma\left(F_2\right)$ is non-empty, i.e., $\sigma\left(F_1\right) \cap \sigma\left(F_2\right) \neq \emptyset$. Conversely, if $\sigma\left(F_1\right) \cap \sigma\left(F_2\right)=\emptyset$, then $F_1$ and $F_2$ are said to be \textit{distinguishable}. We also refer to the pair $\left(F_1, F_2\right)$ as an \textit{indistinguishable pair} when $\sigma\left(F_1\right) \cap \sigma\left(F_2\right) \neq \emptyset$, and as a distinguishable pair when $\sigma(F_1) \cap \sigma(F_2)=\emptyset$.

\subsection{Preliminaries}

For a graph $G$, let $V(G)$, $E(G)$, and $\overline{G}$
denote the set of vertices, the set of edges, the
complement of $G$, respectively. A subgraph $H$ of $G$ is a graph
with $V(H)\subseteq V(G)$, $E(H)\subseteq E(G)$, and the endpoints
of every edge in $E(H)$ belonging to $V(H)$. For any subset $X$ of
$V(G)$, let $G[X]$ denote the subgraph induced by $X$; similarly,
for any subset $F$ of $E(G)$, let $G[F]$ denote the subgraph induced
by $F$. We use $G-X$ to denote the subgraph of $G$ obtained by
removing all the vertices of $X$ together with the edges incident
with them from $G$; similarly, we use $G-F$ to denote the subgraph
of $G$ obtained by removing all the edges of $F$ from $G$. For two subsets $X$ and $Y$
of $V(G)$ we denote by $E_G[X,Y]$ the set of edges of $G$ with one
end in $X$ and the other end in $Y$. If $X=\{x\}$, we simply write
$E_G[x,Y]$ for $E_G[\{x\},Y]$. The {\it degree}\index{degree} of a
vertex $v$ in a graph $G$, denoted by $deg_G(v)$, is the number of
edges of $G$ incident with $v$. Let $\delta(G)$ and $\Delta(G)$ be
the minimum degree and maximum degree of the vertices of $G$,
respectively. The set of neighbors of a vertex $v$ in a graph $G$ is
denoted by $N_G(v)$. The {\it union} $G\cup H$ of two graphs $G$ and
$H$ is the graph with vertex set $V(G)\cup V(H)$ and edge set
$E(G)\cup E(H)$. Let $F_1$ and $F_2$ be two distinct subsets of $V$, then their \textit{symmetric difference} is $F_1 \Delta F_2=\left(F_1- F_2\right) \cup\left(F_2 - F_1\right)$. The \emph{connectivity} of a graph $G$, denoted by $\kappa(G)$, is
the minimal number of vertices whose removal from $G$ produces a
disconnected graph or only one vertex.

\subsection{Diagnosis models}

The identification of faulty nodes is referred to as the diagnosis of the system. In 1967, Preparata et al. \cite{PMC67} put forward a model and a framework named system-level diagnosis. This framework enabled the system to automatically test the processors by itself. System-level diagnosis is widely recognized as an alternative to traditional circuit-level testing in large-scale multiprocessor systems. In the over four-decade period subsequent to this groundbreaking work, numerous terms related to system-level diagnosis have been defined, and different models such as the PMC, BGM, and comparison models have been explored in the literature \cite{PMC67,MM81,MM78,BGM76,FS80}. Among the proposed models, two prominent diagnosis models, namely the Preparata, Metze, and Chien (PMC) model \cite{PMC67} and the Maeng and Malek (MM) model \cite{MM81}, have been extensively utilized.

In the PMC model, for any two nodes $u$ and $v$, if there exists a connecting link between them, node $u$ has the capacity to test node $v$. Here, $u$ is designated as the \textit{tester} and $v$ is named the \textit{tested node}. When a test is carried out by a non-faulty tester, the test result is 1 in the case where the tested node is faulty and 0 when the tested node is non-faulty. However, when the tester itself is faulty, the test result cannot be trusted.

In the MM model, a particular node, known as a comparator, dispatches an identical task to its two adjacent neighbors and then evaluates their responses. The comparison of nodes $u$ and $v$ carried out by node $w$ is symbolically represented as $(u, v)_w$. When a fault-free comparator detects a discrepancy in the comparison, it signals the presence of a faulty node. Conversely, when the comparator is itself faulty, the result of the comparison cannot be relied upon. One of the primary strengths of this model lies in its simplicity, as it is straightforward to conduct a comparison between a pair of nodes within multiprocessor systems. This method is appealing because it does not necessitate any additional hardware, and both transient and permanent faults can be detected prior to the completion of the comparison program. A research paper by Sengupta and Dahbura \cite{SD92} uncovered crucial characteristics of a diagnosable system within the context of this model. It proposed a specialized variant of the MM model, namely the $\mathrm{MM}^*$ model. In the $\mathrm{MM}^*$ model, if $u$ and $v$ are neighbors of $w$ in the system, then $w$ is obligated to perform the comparison $(u, v)_w$. The paper also introduced a polynomial-time algorithm for identifying faulty nodes in a general system under the MM* model, provided that the system is diagnosable. The $\mathrm{MM}^*$ model has been employed in studies such as \cite{CCC07,CH11,LL17,HL11,HK13,YLMLQZ15}.

The PMC model (Preparata, Metze, and Chien model) is a popular system-level diagnosis model proposed by Preparata et al. \cite{PMC67}. In the PMC model, every vertex can test whether each of its neighboring vertices is faulty. We will only study the PMC model in this paper.

\begin{theorem}[Dahbura and Masson, 1984, \cite{DM84}]\label{th-DM84}
Let $G$ be a graph. For each paired distinct sets $F_1, F_2 \subset V, F_1$ and $F_2$ are distinguishable under the PMC model if and only if there exist $u \in F_1 \triangle F_2$ and $v \in \overline{F_1 \cup F_2}$ such that $u v \in E$; see Fig. \ref{Figure1}.
\end{theorem}

This theorem is of vital importance for the study of diagnosis under the PMC model and will be used multiple times in this paper. See \cite{CLH12,CH18,LL17,
PLTH12,
YLMLQZ15} for more papers on the diagnosability under the PMC model.

\begin{figure}[!hbpt]
\begin{center}
\includegraphics[scale=0.6]{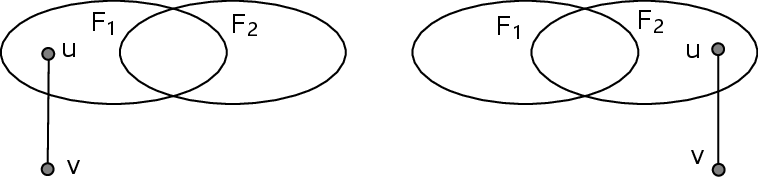}
\end{center}
\caption{Graphs for Theorem \ref{th-DM84}.}\label{Figure1}
\end{figure}

\subsection{$g$-good-neighbor connectivity and diagnosability}

The $g$-good-neighbor connectivity has been an object of interest for many
years, and it was firstly introduced by
Latifi et al. \cite{LHN94}.

\begin{definition}
Let $G=(V,E)$ be connected. A fault set $F\subseteq V$ is called a
\emph{$g$-good-neighbor faulty set} if $|N(v)\cap (V-F)|\geq g$ for
every vertex $v$ in $V-F$, that is, $\delta(G-F)\geq g$, that is, $deg_G(v)\geq g$ for any $v\in V(G-F)$.
\end{definition}

\begin{definition}
A \emph{$g$-good-neighbor cut} of $G$ is
a $g$-good-neighbor faulty set $F$ such that $G-F$ is disconnected.
\end{definition}

\begin{definition}
The minimum cardinality of $g$-good-neighbor cuts is said to be the
\emph{$g$-good-neighbor connectivity} of $G$, denoted by
$\kappa^{g}(G)$.
\end{definition}

\begin{definition}
A connected graph $G$ is said to be
\emph{$g$-good-neighbor connected} if $G$ has a $g$-good-neighbor
cut.
\end{definition}

For more research on $g$-good-neighbor connectivity,
we refer
to \cite{LL17, LXWZ18, PLTH12, WX17, WX18}.

The range of the integer $g$ was determined in \cite{WMHW20}.
\begin{proposition}{\upshape \cite{WMHW20}}\label{pro-g}
Let $g$ be a non-negative integer. If $G$ has $g$-good-neighbor
connectivity and $n$ vertices, then
$$
0\leq g\leq \min\left\{\Delta(G),\left\lfloor
\frac{n-3}{2}\right\rfloor\right\}
$$
and
$$
e(G)\leq {n\choose 2}-(g+1)^2.
$$
\end{proposition}

Peng et al. \cite{PLTH12} proposed the $g$-good-neighbor conditional diagnosability by claiming that for every fault-free vertex in a system, it has at least $g$ fault-free neighbors.

\begin{definition}
A system $G=(V(G), E)$ is
\emph{$g$-good-neighbor conditional $t$-diagnosable} if each distinct pair of $g$-good-neighbor conditional faulty sets $F_1$ and $F_2$ of $V$ with $\left|F_1\right| \leq t,\left|F_2\right| \leq t$ are distinguishable.
\end{definition}

\begin{definition}
The \emph{$g$-good-neighbor conditional diagnosability} $t^g(G)$ of a graph $G$ is the maximum value of $t$ such that $G$ is $g$-good-neighbor conditionally $t$-diagnosable.
\end{definition}

For some papers, $t^g(G)$ is also called the $g$-good-neighbor conditional diagnosability.

Lin et al. \cite{LXWZ18} investigated the $g$-good-neighbor conditional
diagnosability of arrangement graphs.
Xu et al. \cite{XZL17,XLZHG17} studied the g-good-neighbor diagnosability of complete cubic networks and complete cubic networks. Yuan et al. \cite{YLMLQZ15,YLQZL16}
studied the $g$-good-neighbor diagnosability of the $k$-ary $n$-cube ($k \geq 3$) under the PMC model and the MM${}^*$ model.
Ren and Wang \cite{RW17} gave the $g$-good-neighbor diagnosability of locally twisted cubes.

\subsection{Cartesian product networks}

The \emph{Cartesian product} of two graphs $G$ and $H$, written as
$G\Box H$, is the graph with vertex set $V(G)\times V(H)$, in which
two vertices $(u,v)$ and $(u',v')$ are adjacent if and only if
$u=u'$ and $(v,v')\in E(H)$, or $v=v'$ and $(u,u')\in E(G)$. By symmetry, we have $G\Box H=H\Box G$.

\begin{observation}\label{obs-1}
For $X_1,X_2\subseteq V(G)$ and $Y_1,Y_2\subseteq V(H)$ with $X_1\cap X_2=\emptyset$ and $Y_1\cap Y_2=\emptyset$, we have
$$
(X_1\cup X_2)\Box (Y_1\cup Y_2)=(X_1\Box Y_1)\cup (X_1\Box Y_2)\cup (X_2\Box Y_1)\cup (X_2\Box Y_2).
$$
\end{observation}

Let $G$ and $H$ be two connected graphs with $V(G)=\{u_1, u_2, \ldots, u_n\}$ and $V(H)=\left\{v_1, v_2, \ldots, v_m\right\}$, respectively. Then $V(G\square H)=\{(u_i, v_j) \mid 1 \leq i \leq n, 1 \leq j \leq m\}$. For $v \in V(H)$, we use $G(v)$ to denote the subgraph of $G\square H$ induced by the vertex set $\{(u_i, v) \mid 1 \leq i \leq n\}$. Similarly, for $u \in V(G)$, we use $H(u)$ to denote the subgraph of $G\square H$ induced by the vertex set $\{(u, v_j) \mid 1 \leq j \leq m\}$.

Some well-known interconnection networks, e.g., hypercubes, meshes, tori, $k$-ary $n$-cubes, and generalized hypercubes, are product networks
\cite{AS00,BBKA95,BA84}.

Chang et al. \cite{CLTH04} studied the
diagnosability of Cartesian product of two regular networks. Hsieh and Chen \cite{HC08} determined the strong diagnosability of several widely used multiprocessor systems, such as hypercubes, mesh-connected $k$-ary $n$-cubes, torus-connected $k$-ary $n$-cubes, and hyper-Petersen networks.

There are two motivations for this paper. As we can see, there is much research on $g$-good-neighbor connectivity diagnosability. But there are no results on the existence of these two parameters. There are several papers on the diagnosability of two product networks. But there are no results on the diagnosability of Cartesian product of two general networks.

Motivated by studying the existence of $g$-good-neighbor connectivity conditional diagnosability and the Cartesian product of general networks, we study the gc numbers (see Definition \ref{def3-1}), $g$-good-neighbor connectivity, and $g$-good-neighbor diagnosability of general works and general Cartesian networks. The outline of this paper can be seen in Fig. \ref{outline}. As we can see, the existence of $\kappa^{g}(G)$ is very important for studying $t^{g}(G)$, and gc number is a key parameter to compute the $g$-good-neighbor diagnosability.
\begin{itemize}
\item In Section \ref{Section2}, we give two necessary
and sufficient conditions that $t^{g}(G)$ exists
(Theorems \ref{th-condition} and \ref{th-condition-D}).

\item In Section \ref{Section3}, we introduce a new concept called $g$-good neighbour cut-component number
    (gc number for short), denoted by $c^g(G)$, of a graph $G$, and give sharp lower and upper bounds of  $c^g(G)$ of a general graph $G$ in terms of the $g$-good neighbor connectivity $\kappa^g(G)$ (Theorems \ref{th-gc-Lower} and \ref{th-gc-Upper}). These two theorems provide a formula to compute $c^g(G)$ (Corollary \ref{cor-formula-c}).

\item In Section \ref{Section4}, we show the close relation between $t^{g}(G)$ and $c^{g}(G)$ of a graph $G$.
We give an upper bound of $t^{g}(G)$ in terms of gc number (Theorem \ref{th-Upper-gc}), and then derive a formula to compute $t^{g}(G)$ (Theorem \ref{th-Upper-gc-condition}).

\item In Section \ref{Section5}, we derive a sharp upper bound of $t^{g}(G\Box H)$ for two general graphs $G$ and $H$ (Theorem \ref{th-Upper-D}). By Theorem \ref{th-Upper-gc-condition}, we get a formula to compute  $t^{g}(G\Box H)$ for two general graphs $G$ and $H$ (Corollary \ref{th-cor}).

\item As applications, we get the exact values or bounds for the gc numbers and $g$-good-neighbor diagnosability of grid, torus networks and generalized cubes in Section \ref{Section6}.
\end{itemize}

\section{Necessary and sufficient conditions}\label{Section2}

Let $G$ be a graph obtained by a clique $K_g$ and the graph $H$ by adding two vertices $u,w$ such that $H$ is a graph with maximum degree $\Delta(G)<g-2$, $g\geq 3$ and $|V(H)|=r$; see Fig. \ref{Figure3}.
\begin{remark}\label{rem1-1}
One can see that $t^{g}(G) \ (g\geq 3)$ may not exist. Suppose that $F$ is a $g$-good-neighbor conditional faulty set. Then one of the following holds.
\begin{itemize}
\item $V(H)\subseteq F$ and $K_g\cap F=\emptyset$;
\item $u\in F$ and $w\in F$ do not hold.
\end{itemize}
Then $F=V(H)$ or $F=V(H)\cup \{u\}$ or $F=V(H)\cup \{w\}$, and there is no other  $g$-good-neighbor conditional faulty sets in $G$.
\begin{figure}[!hbpt]
\begin{center}
\includegraphics[scale=0.7]{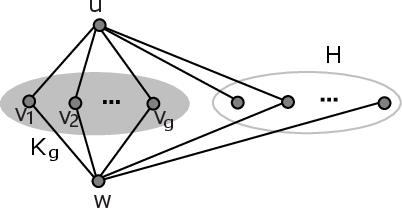}
\end{center}
\begin{center}
\caption{Graphs for Remark \ref{rem1-1}}\label{Figure3}
\end{center}
\end{figure}

For any two different $g$-good-neighbor conditional faulty sets $F_1,F_2\subseteq V(G)$, $F_1=V(H)\cup \{u\}$ or $F_1=V(H)\cup \{w\}$ or $F_2=V(H)\cup \{u\}$ or $F_2=V(H)\cup \{w\}$. Then $(F_1,F_2)$ is a $g$-good-neighbor distinguishable pair in $G$ under the PMC model. For any integer $s\geq r+1$, we have $|F_i|\leq s$ for $i=1,2$, and $(F_1,F_2)$ is a $g$-good-neighbor distinguishable pair in $G$ under the PMC model.
From the arbitrariness of $s$, we have there is no maximum value of $t^{g}(G)$ and so $t^{g}(G)$ can not exist.
\end{remark}

In this subsection, we give two necessary and sufficient conditions that $t^{g}(G)$ exists.

\begin{theorem}\label{th-condition}
Let $G$ be a graph, and let $s$ be a non-negative integer. Then
$t^{g}(G)=s$ if and only if
\begin{center}
$(\mathcal{C})$: there exist two indistinguishable $g$-good-neighbor conditional faulty sets $F_1,F_2$ such that $|F_1|=s+1$ and $|F_2|\leq s+1$. ~~~~~~~~~~~~~~~~~~~~~~~~~~~~~~~~~~~~~~~~~~~~~~~~~~~~~~
\end{center}
\end{theorem}
\begin{proof}
Suppose that $t^g(G)=s$. If there does not exist two indistinguishable $g$-good-neighbor conditional faulty sets $F_1,F_2$ such that $|F_1|=s+1$ and $|F_2|\leq s+1$, then $t^{g}(G)=s+1$, which contradicts the fact that $t^{g}(G)=s$.

Suppose that there exist two indistinguishable $g$-good-neighbor conditional faulty sets $F_1,F_2$ such that $|F_1|=s+1$ and $|F_2|\leq s+1$. We choose $F_1,F_2$ such that $s$ is minimum. Then $t^g(G)\leq |F_1|-1=s$. Let $t^{g}(G)=t$. Then there exist two indistinguishable $g$-good-neighbor conditional faulty sets $F_1',F_2'$ such that $|F_1'|=t+1$ and $|F_2'|\leq t+1$.
If $t<s$, then $|F_1'|=t+1<s+1$, which contradicts the fact that $s$ is minimum. Then
$t^{g}(G)=t=s$.
\end{proof}

\begin{corollary}
Let $g$ be a non-negative integer. If $G$ has $g$-good-neighbor conditional diagnosability, then
$$
0 \leq g \leq \Delta(G).
$$
\end{corollary}

\begin{theorem}\label{th-condition-D}
Let $G$ be a graph, and let $s$ be a non-negative integer. Then
$t^{g}(G)$ exists if and only if
\begin{center}
$(\mathcal{D})$: $\kappa^g(G)$ exists or there exist two indistinguishable $g$-good-neighbor conditional faulty sets $F_1,F_2$ such that $F_1\cup F_2=V(G)$. ~~~~~~~~~~~~~~~~~~~~~~~~~~~~~~~~~~~~~~~~~~~~~~~~~~~~~~
\end{center}
\end{theorem}
\begin{proof}
Suppose that the condition $(\mathcal{D})$ holds. Suppose that $\kappa^g(G)$ exists. Then there exists a cutset $X\subseteq V(G)$ such that the minimum degree of each component of $G-X$ is at least $g$. Let $C_1,C_2,...,C_r$ be the connected components of $G-X$. Without loss of generality, let $|C_1|\leq |C_2|\leq \cdots \leq |C_r|$. Let $F_1=C_1\cup X$ and $F_2=X$. Then $(F_1,F_2)$ is a  indistinguishable $g$-good-neighbor conditional pair. From Theorem \ref{th-condition}, $t_g(G)$ exists.
Suppose that there exist two indistinguishable $g$-good-neighbor conditional faulty sets $F_1,F_2$ such that $F_1\cup F_2=V(G)$. Then $(F_1,F_2)$ is a  indistinguishable $g$-good-neighbor conditional pair. From Theorem \ref{th-condition}, $t^g(G)$ exists.

Suppose that $t^g(G)$ exists. From Theorem \ref{th-condition}, there exist two indistinguishable $g$-good-neighbor conditional faulty sets $F_1,F_2$ such that $|F_1|=s+1$ and $|F_2|\leq s+1$.

If $F_1\cup F_2=V(G)$, then the result follows.

Suppose that $F_1\cup F_2\neq V(G)$. Then $\overline{F_1\cup F_2}\neq \emptyset$.
\begin{claim}\label{claim-D}
$F_1\cap F_2\neq \emptyset$.
\end{claim}
\begin{proof}
Otherwise, we assume that $F_1\cap F_2= \emptyset$. Since $F_1\cup F_2\neq V(G)$, it follows that
there exists a vertex $u\in \overline{F_1\cup F_2}$ adjacent to a vertex $u'\in F_1\cup F_2$, since $G$ is connected. From Theorem \ref{th-DM84}, $(F_1,F_2)$ is a distinguishable $g$-good-neighbor conditional pair, a contradiction.
\end{proof}

From Claim \ref{claim-D}, we have
$F_1\cap F_2\neq \emptyset$. Since $(F_1,F_2)$ is an indistinguishable $g$-good-neighbor conditional pair, it follows that there does not exist an edge from $\overline{F_1\cup F_2}$ to $F_1\triangle F_2$. Then $F_1\cap F_2$ is a $g$-good-neighbor cutset of $G$. Then $\kappa^g(G)$ exists.
\end{proof}

\section{$g$-good neighbour cut-component numbers}\label{Section3}

Let $C_1,C_2,...,C_r$ be the connected components of $G-X$, where $X$ is a $g$-good-neighbor cutset of $G$ with $|X|\geq \kappa^g(G)$. Let $\mathcal{X}=\{X\,|\,|X|\geq \kappa^g(G)\}$.

We assume that there exists an integer $s$ such that
\begin{itemize}
    \item For each $C_i \ (1\leq i\leq s)$, there exists a partition of $C_i$: two vertex subsets $A_i,B_i$ such that $A_i\cup B_i=C_i$ with $\delta(G[A_i])\geq g$ and $\delta(G[B_i])\geq g$ for $j=1,2$. In this partition, we find such two subsets $A_i,B_i$ such that $|A_i|-|B_i|$ is minimum, where $0\leq |A_i|-|B_i|\leq |C_i|-2g-2$ we assume that $|A_i|\geq |B_i|$.

    \item For each $C_i \ (s+1\leq i\leq r)$, there does not exist such a partition of $C_i$.
\end{itemize}

Without loss of generality, let $|A_1|\leq |A_2|\leq \cdots \leq |A_s|$ and $|C_{r}|\leq |C_{r-1}|\leq \cdots \leq |C_{s+1}|$.
Let $\mathcal{A}(X)=\{A_1,A_2,...,A_s\}$ and $\mathcal{C}(X)=\{C_{s+1},C_{s+2},...,C_{r}\}$.

\begin{definition}\label{def3-1}
The \textit{$g$-good neighbour cut-component number} (\textit{gc number} for short) is defined as $c^g(G)=\min\{\min\{|A_1|,|C_r|\}+|X|\,:\,X\in \mathcal{X},$~$A_1\in \mathcal{A}(X)$,~$C_r\in \mathcal{C}(X)$~and $\delta(G-X) \geq g\}$.
\end{definition}

Let $\mathcal{G}_1$ be a graph class such that $G\in \mathcal{G}_1$ if and if only if there exists a $g$-good neighbor cut $X$ of $G$ such that
$c^g(G)=|C_r|+|X|$.
Let $\mathcal{G}_2$ be a graph class such that $G\in \mathcal{G}_2$ if and if only if there exists a $g$-good neighbor cut $X$ of $G$ such that
$c^g(G)=|A_1|+|X|$. If $G\in \mathcal{G}_1$ and $G\in \mathcal{G}_2$, then we let $G\in \mathcal{G}_1$.

\begin{algorithm}[H]
\footnotesize
\caption{The algorithm   determining $c^g(G)$}%算法名字
\label{algorithm:Caculate Stain tree}
\LinesNumbered %要求显示行号
\KwIn{A graph $G$
and $X\subseteq V(G)$ with $|X|\geq \kappa^g(G)$;}% 输入参数
\KwOut{$c^g(G)$ ;}%输出
$C \gets |V(G)|$\;%\;用于换行
$a \gets |V(G)|$\;
$d \gets |V(G)|$\;
$t \gets 0$\;
$A_1 \gets \emptyset$\;
\For{each vertex set $X\subseteq V(G)$
 with $|X|\geq \kappa^g(G)$}{

$\mathcal{C}$ $ \gets$ the  connected components set of $G\setminus X$\;
$\mathcal{T} \gets$ $\mathcal{C}$ \;
\For{each vertex set $C\subseteq \mathcal{C}$
}{
\For{each vertex set $A\subseteq C$
}{
\If{$A$ such with $\delta(G[A])\geq g$ ,$\delta(G[C\setminus A])\geq g$ and $|A|\geq |C\setminus A|$
}{
$d_1\gets$$|A|-|C\setminus A|$ \;

\If{$d_1 < d$}{
$d$ $ \gets$ $d_1$ \;
$A_1$ $ \gets$ $A$ \;
 }
 $\mathcal{T}  \gets \mathcal{T} \setminus C$\;
    }
}
\If{$|A_1| < a$}{
 $a$ $ \gets$ $|A_1|$\;
 }
}
$t$ $ \gets$ the minimum cardinal vertex subset in $\mathcal{T}$ \;

$T_1$ $ \gets$  $|X|+\min\{t,a\}$\;

\If{$T_1 < C$}{
 $C$ $ \gets$ $T_1$\;
 }
}
$c^g(G) \gets  C$ \;
\Return $c^g(G)$ \;
\end{algorithm}

We can give a lower bound of $c^g(G)$ in terms of $\kappa^g(G)$.

\begin{theorem}\label{th-gc-Lower}
Let $G$ be a connected graph, and let $g$ be an integer with $g\geq 0$. If $\kappa^g(G)$ exists, then
$$
c^g(G)\geq \kappa^g(G)+g+1.
$$
Moreover, the equality holds if and only if there exists a cutset $X=\kappa^g(G)$ such that the minimum degree of each connected component of $G-X$ is at least $g$ and one of the following condition holds.
\begin{itemize}
    \item one of connected components is a clique $K_{g+1}$.

    \item one of connected components is obtained from two cliques $K_{g+1}$ by adding at least one edge between them.
\end{itemize}
\end{theorem}
\begin{proof}
From the definition of $c^g(G)$, there exists a cutset $X\subseteq V(G)$ with $|X|\geq \kappa^g(G)$ and a vertex subset $C_r$ or $A_1$ such that
\begin{equation}\label{eq-1}
c^g(G)=|X|+\min\{|C_1|,|A_1|\}\geq \kappa^g(G)+g+1
\end{equation}
and  $|N_G(v) \cap(V(G)-X)| \geq g$ for every vertex $v$ in $V(G)-X$.

Suppose that there exists a cutset $X=\kappa^g(G)$ such that the minimum degree of each connected component of $G-X$ is at least $g$ and one of connected components is a clique $K_{g+1}$. Then $
c^g(G)\leq |X|+|K_{g+1}|=\kappa^g(G)+g+1
$, and hence $
c^g(G)= \kappa^g(G)+g+1
$ by this theorem.

Suppose that there exists a cutset $X=\kappa^g(G)$ such that the minimum degree of each connected component of $G-X$ is at least $g$ and one of connected components, say $C_j$, is obtained from two cliques $K_{g+1}$ by adding at least one edge between them. Clearly, we have $A_j=V(K_{g+1})$ and $B_j=V(K_{g+1})$
Then $
c^g(G)\leq |X|+|V(K_{g+1})|=\kappa^g(G)+g+1
$, and hence $
c^g(G)= \kappa^g(G)+g+1
$ by this theorem.

Conversely, we suppose that $
c^g(G)= \kappa^g(G)+g+1
$.
From the definition of $c^g(G)$, there exists a cutset $X\subseteq V(G)$ with $|X|\geq \kappa^g(G)$ and a vertex subset $A_1$ or $C_r$ such that $|C_r|+|X|=\kappa^g(G)+g+1$ or $|A_1|+|X|=\kappa^g(G)+g+1$ and the minimum degree of each connected component of $G-X$ is at least $g$. Then $|X|\geq  \kappa^g(G)$, and $|C_r|\geq g+1$ or $|A_1|\geq g+1$, and hence $|X|=\kappa^g(G)$, and $|C_r|=g+1$ or $|A_1|=g+1$. Since the minimum degree of $C_r$ is at least $g$, it follows that $C_r=K_{g+1}$. Since $\delta(G[A_1])\geq g$ and $\delta(G[B_1])\geq g$, it follows that the connected component containing $A_1$ is obtained from two cliques $K_{g+1}$ by adding at least one edge between them.
\end{proof}

\begin{corollary}\label{pro1}
Let $T$ be a tree.
If $g=0$, then
$c^g(T)=2$.
\end{corollary}

We can give an upper bound of $c^g(G)$ in terms of $\kappa^g(G)$.

\begin{theorem}\label{th-gc-Upper}
Let $G$ be a connected graph, and let $g$ be an integer with $g\geq 0$. If $\kappa^g(G)$ exists, then
$$
c^g(G)\leq  \kappa^g(G)+\min\{|C_r|,|A_1|\},
$$
where $X$ is a $g$-good neighbor cutset of $G$ with $|X|=\kappa^g(G)$ and $C_r$ is a minimal connected component of $G-X$ in $\mathcal{C}$ and $A_1$ is a minimal vertex subset of $G-X$ in $\mathcal{A}$.

Moreover, the bound is sharp.
\end{theorem}
\begin{proof}
From the definition of $\kappa^g(G)$, there exists a subset $X\subseteq V(G)$ with $|X|=\kappa^g(G)$ such that $G\setminus X$ is not connected  and the minimum degree of each connected component of $G-X$ is at least $g$.
From the definition of $c^g(G)$, we have $c^g(G)\leq |X|+\min\{|C_r|,|A_1|\}=\kappa^g(G)+\min\{|C_r|,|A_1|\}$.
\end{proof}

To show the sharpness of Theorem \ref{th-gc-Upper}, we consider the following example.
\begin{example}
Let $H_n$ be a graph of order $n$ obtained from two $g$-regular connected graphs $D_{\lfloor (n-r)/2\rfloor}$ and $D_{\lceil (n-r)/2\rceil}$ by adding new vertices $u_1,u_2,...,u_{r}$ and new edges in $\{u_iu\,|\,1\leq i\leq r\}\cup \{u_iv\,|\,1\leq i\leq r\}$, where $r=\kappa^g(H_n)$, $u\in V(D_{\lfloor (n-r)/2\rfloor})$, $v\in V(D_{\lceil (n-r)/2\rceil})$, and $3\leq g\leq \lceil (n-r)/2\rceil-1$. For any $g$-good neighbor cut $X\subseteq V(H_n)$, we have $\{u_i\,|\,1\leq i\leq r\}\subseteq X$, since $g\geq 3$ and $deg_{H_n}(u_i)=2$. Since $D_{\lfloor (n-r)/2\rfloor}$ and $D_{\lceil (n-r)/2\rceil}$ are two $g$-regular connected graphs, it follows that $X\cap V(D_{\lfloor (n-r)/2\rfloor})=\emptyset$ and $X\cap V(D_{\lceil (n-r)/2\rceil})=\emptyset$. Then $X=\{u_i\,|\,1\leq i\leq r\}$, and hence
$c^g(H_n)=r+\lfloor (n-r)/2\rfloor=\kappa^g(H_n)+|D_{\lfloor (n-r)/2\rfloor}|$, and hence the bound is sharp.
\end{example}

The following corollary is immediate by Theorems \ref{th-gc-Upper} and \ref{th-gc-Lower}.
\begin{corollary}\label{cor-formula-c}
Let $G$ be a connected graph, and let $g$ be an integer with $g\geq 0$. If $\kappa^g(G)$ exists and there exists a minimum
$g$-good-neighbor
cut $X$ and a vertex subset $C_r$ or $A_1$ with $|C_r|=|A_1|=g+1$, then
$$
c^g(G)=\kappa^g(G)+g+1.
$$
\end{corollary}

\section{Results for $g$-good-neighbor diagnosability}\label{Section4}

We now give an upper bound of $t^{g}(G)$ in terms of gc numbers and order.
\begin{theorem}\label{th-Upper-gc}
Let $G$ be a connected graph with $n$ vertices. If $\kappa^g(G)$ exists, then
$$
t^{g}(G)\leq \min\{c^g(G)-1,n-g-2\}.
$$
Moreover, the bound is sharp.
\end{theorem}
\begin{proof}
From the definition of $c^g(G)$, there exists a subset $X\subseteq V(G)$ with $|X|\geq \kappa^g(G)$, and a connected component $C_r\in \mathcal{C}$ or a vertex subset $A_1\in \mathcal{A}$ of $G\setminus X$ such that $|C_r\cup X|=c^g(G)$ or $|A_1\cup X|=c^g(G)$, and the minimum degree of each connected component of $G-X$ is at least $g$. Let $F_1=X\cup C_r$ and $F_2=X$ if $|C_r\cup X|=c^g(G)$, and let $F_1=X\cup A_1$ and $F_2=X\cup B_1$ if $|A_1\cup X|=c^g(G)$.
Clearly,
$|F_1|=|X|+|C_r|$ and $|F_2|=|X|$ if $|C_r\cup X|=c^g(G)$, and let $|F_1|=|X|+|A_1|$ and $|F_2|=|X|+|B_1|$ if $|A_1\cup X|=c^g(G)$.
Since there is no edges between each $C_i \ (1\leq i\leq r-1)$
and $F_1 \triangle F_2=C_r$ or there is no edges between each $C_i \ (2\leq i\leq r)$
and $F_1 \triangle F_2=C_1$, it follows that there do not exist two vertices $u,u'$ in $G$ with $u\in C_1=F_1 \triangle F_2$ and $u' \in \overline{F_1 \cup F_2}$ such that $uu'\in E(G)$. From Theorem \ref{th-DM84},
$(F_1,F_2)$ is an indistinguishable pair in $G$ under the PMC model. This means that $t^{g}(G)\leq |F_1|-1=c^g(G)-1.$

Suppose that $
t^{g}(G)\geq n-g-1.$ From Theorem \ref{th-condition}, there exist two indistinguishable $g$-good-neighbor conditional faulty sets $F_1,F_2$ such that $|F_1|=t^{g}(G)+1$ and $|F_2|\leq t^{g}(G)+1$. Then $|F_1|=t^{g}(G)+1\geq n-g$, and hence $|V(G)|-|F_1|\leq g$, a contradiction. So, we have $
t^{g}(G)\leq n-g-2.$
\end{proof}

\begin{example}
Let $F_n$ be a graph of order $n$ obtained from two cliques $K_{g+1}$ and $K_{n-g-1}$ by adding a new edge between them, where $n\geq 2g+3$.
Then $c^g(F_n)=g+2$ and $t^{g}(F_n)=g+1$, and hence the bound is sharp.
\end{example}

\begin{theorem}\label{th-Upper-gc-condition}
Let $G$ be a connected graph with $n$ vertices such that $\kappa^g(G)$ exists. If $c^g(G)\leq \lceil n/2\rceil$, then
\begin{equation}\label{eq-t-c}
t^{g}(G)=c^g(G)-1.
\end{equation}
\end{theorem}
\begin{proof}
Since $\kappa^g(G)$ exists, it follows from Theorem \ref{th-condition-D} that $t^{g}(G)$ exists. Since $\kappa^g(G)$ exists, it follows from Proposition \ref{pro-g} that $g\leq \lfloor
(n-3)/2\rfloor$. From Theorem \ref{th-Upper-gc}, we have $t^{g}(G)\leq \min\{c^g(G)-1,n-g-2\}=c^g(G)-1$.

To show $t^{g}(G)\geq c^g(G)-1$, for any $F_i\subseteq V(G)$ with $|F_i|\leq c^g(G)-1\leq \lceil n/2\rceil-1$, and $F_i$ is a $g$-good-neighbor conditional faulty set, where $i=1,2$, we suppose that $F_1\cap F_2=\emptyset$. Since $|F_1|+|F_2|\leq n-1$, it follows that there exists a vertex $u\in \overline{F_1\cup F_2}$ and there exists a vertex $u'\in F_1\cup F_2$ such that $uu'\in E(G)$. From Theorem \ref{th-DM84}, $(F_1,F_2)$ is a distinguishable $g$-good-neighbor pair in $G$ under the PMC model.

Suppose that $F_1\cap F_2\neq \emptyset$. Then we have the following claim.
\begin{claim}\label{claim}
$(F_1,F_2)$ is a distinguishable $g$-good-neighbor pair in $G$ under the PMC model.
\end{claim}
\begin{proof}
Suppose, to the contrary, that $(F_1,F_2)$ is an indistinguishable $g$-good-neighbor pair in $G$ under the PMC model. Without loss of generality, let $|F_1|\geq |F_2|$. Since $\overline{F_1 \cup F_2}\neq \emptyset$, it follows from Theorem \ref{th-DM84} that there is no edge from $\overline{F_1 \cup F_2}$ to $F_1 \triangle F_2$. Then $F_1\cap F_2$ is a $g$-good-neighbor cutset of $G$. Note that the minimum degree of the subgraph induced by the vertices in $F_1-F_2$ (resp. $F_2-F_1$) is at least $g$.
Therefore, $F_1\triangle F_2$ is a union of connected components of $G-(F_1\cap F_2)$, and hence $|F_1|\geq c^g(G)$, a contradiction.
\end{proof}

From Claim \ref{claim}, $(F_1,F_2)$ is a distinguishable $g$-good-neighbor pair in $G$ under the PMC model, and so $t^{g}(G)\geq c^g(G)-1.$ The result follows.
\end{proof}

\section{Results for $g$-good-neighbor diagnosability of product networks}\label{Section5}

The following result is immediate by Theorem \ref{th-Upper-gc-condition}.
\begin{corollary}\label{th-cor}
Let $G,H$ be two connected graphs
with $n,m$ vertices, respectively.
 such that $\kappa^g(G\Box H)$ exists. If $c^g(G\Box H)\leq \lceil nm/2\rceil$, then
$$
t^{g}(G\square H)=c^g(G\square H)-1.
$$
\end{corollary}
\begin{proof}
From Theorem \ref{th-Upper-gc-condition}, we have
$t^{g}(G\square H)= c^g(G\square H)-1$.
\end{proof}

Under the PMC model, we can derive an upper bound of diagnosability.
\begin{theorem}\label{th-Upper-D}
Let $p,q,g$ be three integers with $p+q=g$.
Let $G,H$ be two connected graphs
with at least $n\geq p+1,m\geq q+1$ vertices, respectively, such that $0\leq p\leq g\leq  \delta(G)+q$ and $0\leq q\leq g\leq \delta(H)+p$.
Suppose that $G,H$ are not complete graphs such that $\kappa^p(G)$ and $\kappa^q(H)$ exist, and $c^p(G)\leq \lceil n/2\rceil$ and $c^q(H)\leq \lceil m/2\rceil$.

$(i)$ If $G\in \mathcal{G}_1$ and $H\in \mathcal{G}_1$, then
$$
t^{g}(G\square H)\leq (t^{p}(G)+1)(t^{q}(H)+1)-1.
$$
Moreover, the bound is sharp.

$(ii)$ If $G\in \mathcal{G}_2$ and $H\in \mathcal{G}_1$, then
$$
\begin{aligned}
t^{g}(G\square H)
&\leq 2(t^{p}(G)+1)(t^{q}(H)+1)-(t^q(H)+1)\kappa^p(G)-(p+1)(q+1)-1.
\end{aligned}
$$
Moreover, the bound is sharp.

$(iii)$ If $G\in \mathcal{G}_1$ and $H\in \mathcal{G}_2$, then
$$
\begin{aligned}
t^{g}(G\square H)
&\leq 2(t^{p}(G)+1)(t^{q}(H)+1)-(t^p(G)+1)\kappa^q(H)-(p+1)(q+1)-1.
\end{aligned}
$$
Moreover, the bound is sharp.

$(iv)$ If $G\in \mathcal{G}_2$ and $H\in \mathcal{G}_2$, then
$$
\begin{aligned}
t^{g}(G\square H)
&\leq (2t^{p}(G)-\kappa^p(G)+2)(2t^{p}(H)-\kappa^q(H)+2)-2(p+1)(q+1)-1.
\end{aligned}
$$
\end{theorem}
\begin{proof}
Let $c^p(G)=c^p$ and $c^q(H)=c^q$. From Theorem \ref{th-Upper-gc-condition}, we have $t^p(G)= c^p-1$ and $t^q(H)= c^q-1$.
Since $\kappa^p(G)$ exists, it follows from the definition of $c^p(G)$ that there exists a cutset $X^p\subseteq V(G)$ with $|X^p|\geq \kappa^p(G)$ and a vertex subset $C_r^p$ or $A_1^p$ such that $|C_r^p|+|X^p|=c^p(G)$ or $|A_1^p|+|X^p|=c^p(G)$ and the minimum degree of the subgraph induced by the vertices of $C_r^p$ or $A_1^p$ in $G-X^p$ is at least $p$.
Since $\kappa^q(H)$ exists, it follows from the definition of $c^q(H)$ that there exists a cutset $X^q\subseteq V(H)$ with $|X^q|\geq \kappa^q(H)$ and a vertex subset $C_r^q$ or $A_1^q$ such that $|C_r^q|+|X^q|=c^q(H)$ or $|A_1^q|+|X^q|=c^q(H)$ and the minimum degree of the subgraph induced by the vertices of $C_r^q$ or $A_1^q$ in $H-X^q$ is at least $q$.
\begin{case}
$G\in \mathcal{G}_1$ and $H\in \mathcal{G}_1$.
\end{case}

In this case, we have $|C_r^p|+|X^p|=c^p(G)$ and  $|C_r^q|+|X^q|=c^q(H)$.
Let
$$
X_1= (C_r^p\Box C_r^q)\cup (C_r^p\Box X^q)\cup (X^p\Box C_r^q)\cup (X^p\Box X^q)
$$
and
$$
X_2=(C_r^p\Box X^q)\cup (C_r^q\Box X^p)\cup (X^q\Box X^p);
$$
see Fig. \ref{Figure4}. From Observation \ref{obs-1}, we have $X_1=((C_r^p\cup X^p) \Box C_r^q)\cup ((C_r^p\cup X^p)\Box X^q)=(C_r^p\cup X^p) \Box (C_r^q\cup X^q)$, and hence
\begin{equation}\label{case1}
|X_1|=|C_r^p\cup X^p||C_r^q\cup X^q|=(|C_r^p|+|X^p|)(|C_r^q|+|X^q|)=c^p(G)c^q(H).
\end{equation}

\setcounter{claim}{0}
\begin{claim}\label{U-claim}
$(X_1,X_2)$ is a $g$-good-neighbor indistinguishable pair in $G\Box H$ under the PMC model.
\end{claim}
\begin{proof}
Otherwise, $(X_1,X_2)$ is a $g$-good-neighbor distinguishable pair in $G\Box H$, that is,
\begin{itemize}
    \item there exist two adjacent vertices $(u,v),(u',v')$ such that $(u,v)\in X_1\triangle X_2$ and $(u',v')\in \overline{X_1\cup X_2}$.

    \item $X_1,X_2$ are $g$-good-neighbor conditional faulty sets in $G\Box H$.
\end{itemize}

From the structure of Cartesian product, $uu'\in E(G)$ and $v=v'$, or $vv'\in E(H)$ and $u=u'$. Without loss of generality, let $uu'\in E(G)$ and $v=v'$. Then
$$
(u,v)\in (X_1\triangle X_2)\cap G(v)=C^p_r\Box \{v\}=C_r^p(v)
$$
and
$$(u',v)\in (\overline{X_1\cup X_2})\cap G(v);
$$
see Fig. \ref{Figure5}.
Then $u\in C_r^p(v)$ and $u'\in G(v)-C_r^p(v)-X^p(v)$, contradiction.
\begin{figure}[!htbp]
\centering
\begin{minipage}{0.45\linewidth}
\vspace{3pt}
\centerline{\includegraphics[width=7.5cm]
{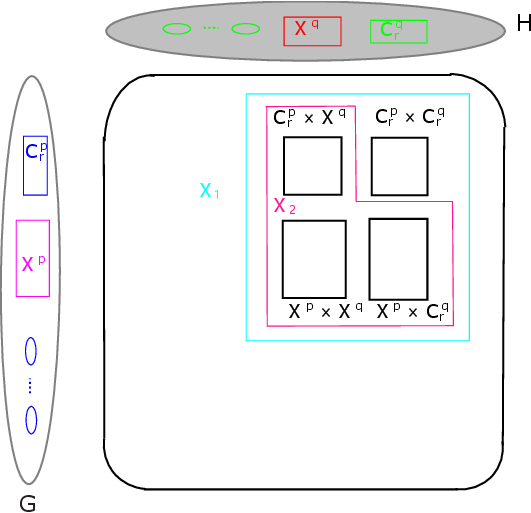}}
\caption{Two subsets $X_1$ and $X_2$, where ``$\times$'' denote the Cartesian product $\Box$.}\label{Figure4}
\end{minipage}
\begin{minipage}{0.45\linewidth}
\vspace{3pt}
\centerline{\includegraphics[width=7.5cm]
{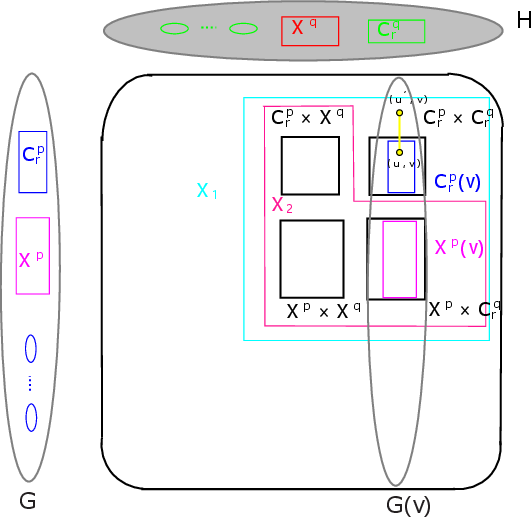}}
\caption{The edge $(u,v)(u',v)$ in $G\Box H$, where ``$\times$'' denote the Cartesian product $\Box$.}\label{Figure5}
\end{minipage}
\end{figure}
\end{proof}

From Claim \ref{U-claim}, $(X_1,X_2)$ is an indistinguishable pair in $G\Box H$ under the PMC model.

For any $(u,v)\in G\Box H$, if $(u,v)\in C^p_r\Box C_r^q$, then $deg_{G\Box H-X_2}((u,v))\geq p+q=g$.
Suppose that $(u,v)\in \overline{X_1}$. If $v\in X^q\cup C_r^q$, then $u\in V(G)-(X^p\cup C_r^p)$ and hence
$deg_{(G\Box H)-X_1}((u,v))\geq p+\delta(H)\geq p+q=g$.
If $v\in V(H)-(X^q\cup C_r^q)$, then
$deg_{(G\Box H)-X_1}((u,v))\geq q+\delta(G)\geq p+q=g$.
Clearly, $X_1$ is a $g$-good-neighbor conditional faulty set in $G\Box H$. Similarly, $X_2$ is a $g$-good-neighbor conditional faulty set in $G\Box H$.

From the above argument and Eqs. (\ref{case1}) and (\ref{eq-t-c}), we have
$$
t^{g}(G\square H)\leq |X_1|-1=c^p(G)c^q(H)-1=(t^{p}(G)+1)(t^{q}(H)+1)-1.
$$

\begin{case}\label{case2}
$G\in \mathcal{G}_2$ and $H\in \mathcal{G}_1$.
\end{case}

In this case, $|A_1^p|+|X^p|=c^p(G)$ and  $|C_r^q|+|X^q|=c^q(H)$.
Let
$$
X_1= (A_1^p\cup B_1^p\cup X^p)\Box (X^q\cup C_r^q)-(B_1^p\Box C_r^q)
$$
and
$$
X_2=(A_1^p\cup B_1^p\cup X^p)\Box (X^q\cup C_r^q)-(A_1^p\Box C_r^q);
$$
see Fig. \ref{Figure66}.
Then
$$
\begin{aligned}
X_1&=(A_1^p\cup B_1^p\cup X^p)\Box (X^q\cup C_r^q)-(B_1^p\Box C_r^q)\\
&=((A_1^p\cup X^p)\Box (X^q\cup C_r^q))\cup (B_1^p\Box (X^q\cup C_r^q)) -(B_1^p\Box C_r^q).
\end{aligned}
$$
Since $|B_1^p|\geq p+1$, $|C_r^q|\geq q+1$, and $c^p(G)=|A_1^p|+|X^p|\geq |B_1^p|+|X^p|$, it follows that $|B_1^p|\leq c^p(G)-|X^p|\leq c^p(G)-\kappa^p(G)$, and
\begin{equation}
\begin{aligned}
|X_1|&\leq (|A_1^p|+|X^p|)(|X^q|+ |C_r^q|)+|B_1^p|(|X^q|+ |C_r^q|)-(p+1)(q+1)\\
&\leq c^p(G)c^q(H)+c^q(H)(c^p(G)-\kappa^p(G))-(p+1)(q+1).
\end{aligned}
\end{equation}

Similarly to the proof of Claim \ref{U-claim}, we can prove that $(X_1,X_2)$ is a $g$-good-neighbor indistinguishable pair in $G\Box H$ under the PMC model.

For any $(u,v)\in G\Box H$,
suppose that $(u,v)\in \overline{X_1}$. Then $(u,v)\in B_1^p\Box C_r^q$ or $(u,v)\in \overline{X_1\cup X_2}$. If $(u,v)\in B_1^p\Box C_r^q$, then $deg_{G\Box H-X_1}((u,v))\geq p+q=g$.
Suppose that $(u,v)\in \overline{X_1\cup X_2}$.
If $v\in X^q\cup C_r^q$, then $u\in V(G)-(X^p\cup C_1^p)$ and hence
$deg_{(G\Box H)-X_1}((u,v))\geq p+\delta(H)\geq p+q=g$.
If $v\in V(H)-(X^q\cup C_r^q)$, then
$deg_{(G\Box H)-X_1}((u,v))\geq q+\delta(G)\geq p+q=g$.

Suppose that $(u,v)\in \overline{X_2}$. Then $(u,v)\in A_1^p\Box C_r^q$ or $(u,v)\in \overline{X_1\cup X_2}$. If $(u,v)\in A_1^p\Box C_r^q$, then $deg_{G\Box H-X_2}((u,v))\geq p+q=g$.
Suppose that $(u,v)\in \overline{X_1\cup X_2}$.
If $v\in X^q\cup C_r^q$, then $u\in V(G)-(X^p\cup C_1^p)$ and hence
$deg_{(G\Box H)-X_1}((u,v))\geq p+\delta(H)\geq p+q=g$.
If $v\in V(H)-(X^q\cup C_r^q)$, then
$deg_{(G\Box H)-X_2}((u,v))\geq q+\delta(G)\geq p+q=g$.

Clearly, $X_1$ is a $g$-good-neighbor conditional faulty set in $G\Box H$. Similarly, $X_2$ is a $g$-good-neighbor conditional faulty set in $G\Box H$.

From the above argument, we have
$$
\begin{aligned}
t^{g}(G\square H)\leq |X_1|-1&\leq c^p(G)c^q(H)+c^q(H)(c^p(G)-\kappa^p(G))-(p+1)(q+1)-1\\
&\leq 2(t^{p}(G)+1)(t^{q}(H)+1)-(t^q(H)+1)\kappa^p(G)-(p+1)(q+1)-1.
\end{aligned}
$$

\begin{case}
$G\in \mathcal{G}_1$ and $H\in \mathcal{G}_2$.
\end{case}

In this case, $|C_r^p|+|X^p|=c^p(G)$ and  $|A_1^q|+|X^q|=c^q(H)$.
Let
$$
X_1= (C_r^p\cup X^p)\Box (X^q\cup A_1^q\cup B_1^q)-(B_1^q\Box C_r^p)
$$
and
$$
X_2=(C_r^p\cup X^p)\Box (X^q\cup A_1^q\cup B_1^q)-(A_1^q\Box C_r^p);
$$
see Fig. \ref{Figure7}.
Then $X_1,X_2\subseteq V(G\Box H)$.
By symmetry, similarly to the proof of Case \ref{case2}, we can prove that $t^{g}(G\square H)\leq 2(t^{p}(G)+1)(t^{q}(H)+1)-(t^p(G)+1)\kappa^q(H)-(p+1)(q+1)-1.$

\begin{case}\label{case4}
$G\in \mathcal{G}_2$ and $H\in \mathcal{G}_2$.
\end{case}

In this case, $|A_1^p|+|X^p|=c^p(G)$ and  $|A_1^q|+|X^q|=c^q(H)$.
Let
$$
\begin{aligned}
X_1
&=(A_1^p\cup B_1^p\cup X^p)\Box (A_1^q\cup B_1^q\cup X^q)-(B_1^p\Box A_1^q)- (A_1^p\Box B_1^q)
\end{aligned}
$$
and
$$
\begin{aligned}
X_2
&=(A_1^p\cup B_1^p\cup X^p)\Box (A_1^q\cup B_1^q\cup X^q)-(B_1^p\Box B_1^q)- (A_1^p\Box A_1^q);
\end{aligned}
$$
see Fig. \ref{Figure7}.
Since $|B_1^p|\geq p+1$, $|C_r^q|\geq q+1$, and $c^p(G)=|A_1^p|+|X^p|\geq |B_1^p|+|X^p|$, it follows that $|B_1^p|\leq c^p(G)-|X^p|\leq c^p(G)-\kappa^p(G)$, and
\begin{equation}
\begin{aligned}
|X_1|&\leq (|A_1^p|+|X^p|+c^p(G)-\kappa^p(G))(|X^q|+ |A_1^q|+c^q(H)-\kappa^q(H))-2(p+1)(q+1)\\
&\leq (2c^p(G)-\kappa^p(G))(2c^q(H)-\kappa^q(H))-2(p+1)(q+1).
\end{aligned}
\end{equation}

Similarly to the proof of Claim \ref{U-claim}, we can prove that $(X_1,X_2)$ is a $g$-good-neighbor indistinguishable pair in $G\Box H$ under the PMC model.
\begin{figure}[!htbp]
\centering
\begin{minipage}{0.45\linewidth}
\vspace{3pt}
\centerline{\includegraphics[width=7.5cm]
{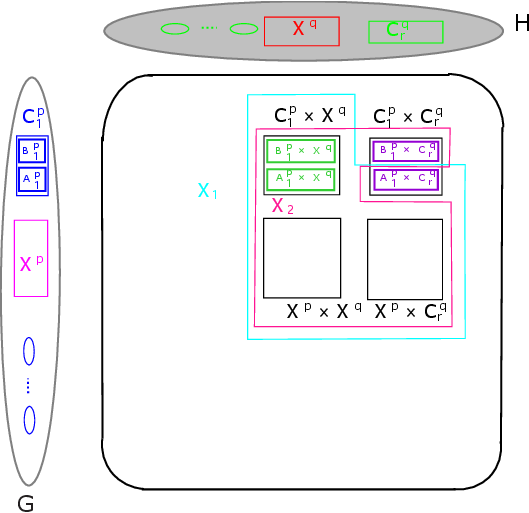}}
\caption{Graphs for Case \ref{case2}, where ``$\times$'' denote the Cartesian product $\Box$.}\label{Figure66}
\end{minipage}
\begin{minipage}{0.45\linewidth}
\vspace{3pt}
\centerline{\includegraphics[width=7.5cm]
{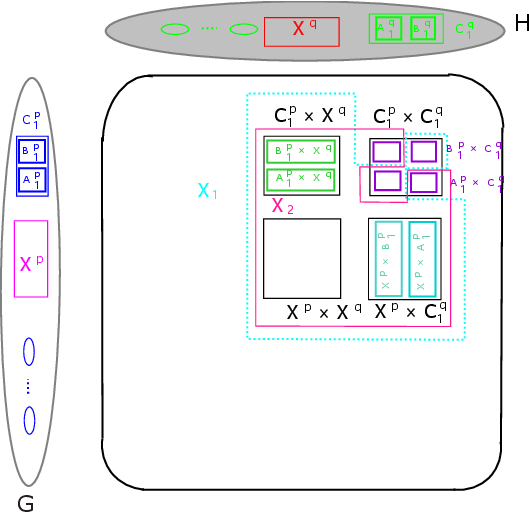}}
\caption{Graphs for Case \ref{case4}, where ``$\times$'' denote the Cartesian product $\Box$.}\label{Figure7}
\end{minipage}
\end{figure}

For any $(u,v)\in G\Box H$,
suppose that $(u,v)\in \overline{X_1}$. Then $(u,v)\in B_1^p\Box A_1^q$ or $(u,v)\in A_1^p\Box B_1^q$ or $(u,v)\in \overline{X_1\cup X_2}$. If $(u,v)\in B_1^p\Box A_1^q$, then $deg_{G\Box H-X_1}((u,v))\geq p+q=g$.
If $(u,v)\in A_1^p\Box B_1^q$, then $deg_{G\Box H-X_1}((u,v))\geq p+q=g$.
Suppose that $(u,v)\in \overline{X_1\cup X_2}$.
If $v\in X^q\cup C_1^q$, then $u\in V(G)-(X^p\cup C_1^p)$ and hence
$deg_{(G\Box H)-X_1}((u,v))\geq p+\delta(H)\geq p+q=g$.
If $v\in V(H)-(X^q\cup C_1^q)$, then
$deg_{(G\Box H)-X_1}((u,v))\geq q+\delta(G)\geq p+q=g$.

Suppose that $(u,v)\in \overline{X_2}$. Then $(u,v)\in A_1^p\Box A_1^q$ or $(u,v)\in B_1^p\Box B_1^q$ or $(u,v)\in \overline{X_1\cup X_2}$. If $(u,v)\in A_1^p\Box A_1^q$, then $deg_{G\Box H-X_2}((u,v))\geq p+q=g$. If $(u,v)\in B_1^p\Box B_1^q$, then $deg_{G\Box H-X_2}((u,v))\geq p+q=g$.
Suppose that $(u,v)\in \overline{X_1\cup X_2}$.
If $v\in X^q\cup C_1^q$, then $u\in V(G)-(X^p\cup C_1^p)$ and hence
$deg_{(G\Box H)-X_1}((u,v))\geq p+\delta(H)\geq p+q=g$.
If $v\in V(H)-(X^q\cup C_1^q)$, then
$deg_{(G\Box H)-X_2}((u,v))\geq q+\delta(G)\geq p+q=g$.

From the above argument, we have
$$
\begin{aligned}
t^{g}(G\square H)\leq |X_1|-1&\leq (2c^p(G)-\kappa^p(G))(2c^q(H)-\kappa^q(H))-2(p+1)(q+1)\\
&\leq (2t^{p}(G)-\kappa^p(G)+2)(2t^{q}(H)-\kappa^q(H)+2)-2(p+1)(q+1)-1.
\end{aligned}
$$
\end{proof}

Let $D_n$ be a graph obtained from two cliques $K_n$ by adding a new vertex $v$ and adding one edge between $v$ and each clique.
\begin{figure}[!hbpt]
\begin{center}
\includegraphics[scale=0.6]{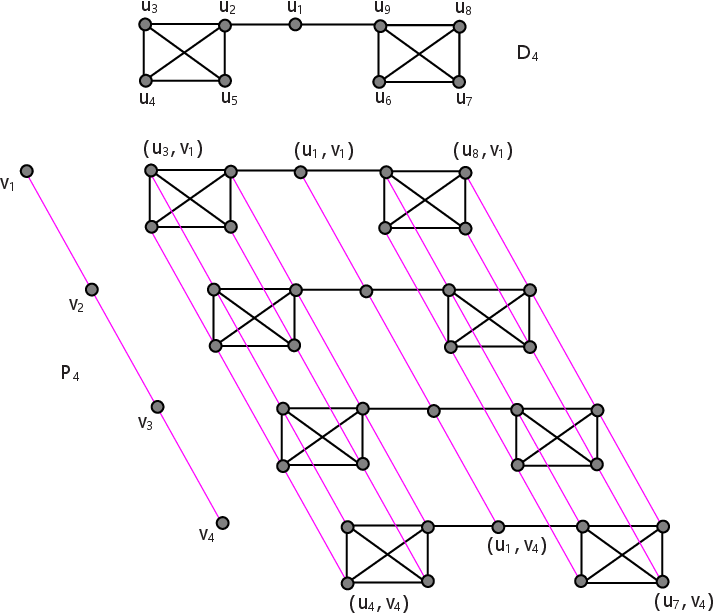}
\end{center}
\begin{center}
\caption{$D_4\Box P_4$}\label{Figure}
\end{center}
\end{figure}

\begin{proposition}\label{pro-DP}
$c^3(D_4\Box P_n)=10$ for $n\geq 4$.
\end{proposition}
\begin{proof}
Let $u_1\in V(D_4)$ such that $D_4-u_1$ contains two cliques of order $4$, say $K_4^1,K_4^2$. Let $V(K_4^1)=\{u_i\,|\,2\leq i\leq 5\}$ and $V(K_8^2)=\{u_i\,|\,6\leq i\leq 9\}$, and let $P_n=v_1v_2...v_n$ be the path of order $n$.

Let $X=\{(u_1,v_1),(u_1,v_2)\}\cup \{(u_i,v_2)\,|\,6\leq i\leq 9\}$. Then $G-X$ is not connected, and $\delta(G-X)\geq 3$. Therefore, $c^3(D_4\Box P_n)\leq |X|+|\{(u_i,v_1)\,|\,6\leq i\leq 9\}|=10$.

Let $G=D_4\Box P_4$, and let $Y=\{(u_1,v_i)\,|\,1\leq i\leq 4\}$. From the definition of $c^3(G)$, there exists a $X\subseteq V(G)$ with $|X|\geq \kappa^3(G)$ and a subset $C_r$ or $A_1$ such that $c^3(G)=|X|+|C_r|$ or $c^3(G)=|X|+|A_1|$.
If $|X|\geq 6$, then $c^3(G)=|X|+|C_r|\geq 6+4=10$ or $c^3(G)=|X|+|A_1|\geq 6+4=10$.
Suppose that $|X|\leq 5$.
If $|Y\cap X|=4$, then $\delta(G-X)\geq 3$, and $c^3(G)=|X|+|C_r|=4+15=19$ or $c^3(G)=|X|+|A_1|=4+\lceil 15/2\rceil=12$.
If $|Y\cap X|=3$, then there is a vertex of degree $2$ in $G-X$, a contradiction. Suppose that $|Y\cap X|=2$. Then $Y\cap X=\{(u_1,v_1),(u_1,v_2)\}$ or $Y\cap X=\{(u_1,v_1),(u_1,v_4)\}$ or $Y\cap X=\{(u_1,v_3),(u_1,v_4)\}$, and hence $|X\cap (V(G)-(X\cap Y)|\leq 3$. To guarantee $\delta(G-X)\geq 3$, $G-X$ is connected, a contradiction. Suppose that $|Y\cap X|=1$. Then $Y\cap X=\{(u_1,v_1)\}$ or $\{(u_1,v_4)\}$, and hence $|X\cap (V(G)-(X\cap Y)|\leq 4$. By symmetry, we assume that $Y\cap X=\{(u_1,v_4)\}$.
To guarantee $\delta(G-X)\geq 3$, we have  $(u_2,v_3),(u_9,v_3)\notin X$, and hence $G-X$ is connected, a contradiction. Suppose that $|Y\cap X|=0$. Then
$|X\cap (V(G)-(X\cap Y)|\leq 5$. To guarantee $\delta(G-X)\geq 3$, $G-X$ is connected, a contradiction. From the above arguments, we have $c^3(G)\geq 10$, and hence $c^3(D_4\Box P_n)\geq c^3(D_4\Box P_4)\geq 10$, and hence $c^3(D_4\Box P_n)=10$.
\end{proof}

The following corollary shows that the bound in Theorem \ref{th-Upper-D} $(i)$ is sharp.
\begin{corollary}\label{cor-9}
$t^3(D_4\Box P_n)=9$ for $n\geq 5$.
\end{corollary}
\begin{proof}
From Proposition \ref{pro-DP} and Corollary \ref{th-cor}, we have $t^3(D_4\Box P_n)=c^3(D_4\Box P_n)-1= 10-1=9$.
\end{proof}

\begin{remark}
From Corollary \ref{cor-9}, we have $
t^{3}(D_4\Box P_n)=9=10-1=(t^{0}(P_n)+1)(t^{3}(D_4)+1)-1$. Note that $p=0$, $q=3$ and $g=3$. This shows that the bound in Theorem \ref{th-Upper-D} $(i)$ is sharp.
\end{remark}

\begin{proposition}\label{pro-DP-2}
$c^3(D_8\Box P_n)=14$ for $n\geq 10$.
\end{proposition}
\begin{proof}
Let $u_1\in V(D_8)$ such that $D_8-u_1$ contains two cliques of order $8$, say $K_8^1,K_8^2$. Let $V(K_8^1)=\{u_i\,|\,2\leq i\leq 8\}$ and $V(K_8^2)=\{u_i\,|\,10\leq i\leq 17\}$, and let $P_n=v_1v_2...v_n$ be the path.

Let $X=\{(u_1,v_1),(u_1,v_2)\}\cup \{(u_i,v_2)\,|\,10\leq i\leq 17\}$. Then $G-X$ is not connected, and $\delta(G-X)\geq 3$. Therefore, $c^3(D_4\Box P_n)\leq |X|+|\{(u_i,v_1)\,|\,10\leq i\leq 17\}|/2=10+4=14$.

Let $G=D_8\Box P_{10}$, and let $Y=\{(u_1,v_i)\,|\,1\leq i\leq 10\}$. From the definition of $c^3(G)$, there exists a $X\subseteq V(G)$ with $|X|\geq \kappa^3(G)$ and a subset $C_r$ or $A_1$ such that $c^3(G)=|X|+|C_r|$ or $c^3(G)=|X|+|A_1|$.
If $|X|\geq 10$, then it follows that $c^3(G)=|X|+|C_r|\geq 10+4=14$ or $c^3(G)=|X|+|A_1|\geq 10+4=14$. Suppose that $|X|\leq 9$.
If $|Y\cap X|\geq 2$, then $|V(K_8^i\Box P_n)\cap X|\leq 7$, $(K_8^i\Box P_n)-X$ is connected, since $\kappa(K_8^i\Box P_n)\geq 8$ for $i=1,2$.. This together with $n\geq 10$, $G-X$ is connected, a contradiction. So $|Y\cap X|\leq 1$.
Suppose that $|Y\cap X|=1$. Then $Y\cap X=\{(u_1,v_j)\}$ for some $j \ (1\leq j\leq 10)$. If $|X\cap (V(G)-(X\cap Y)|< 8$, it follows that $G-X$ is connected, a contradiction. So $|X\cap (V(G)-(X\cap Y)|= 8$. To guarantee $\delta(G-X)\geq 3$, $G-X$ is connected, a contradiction.
Suppose that $|Y\cap X|=0$. Then
$|X\cap (V(G)-Y)|=8$ or $9$. If $|X\cap (V(G)-Y)|=8$, then $G-X$ is connected, a contradiction.
Suppose that $|X\cap (V(G)-Y)|=9$.
To guarantee $\delta(G-X)\geq 3$, $G-X$ is connected, a contradiction.
From the above arguments, we have $c^3(G)\geq 14$, and hence $c^3(D_8\Box P_n)\geq c^3(D_8\Box P_{10})\geq 14$.
\end{proof}

The following corollary shows that the bound in Theorem \ref{th-Upper-D} $(iii)$ is sharp. By symmetry of Cartesian product, the bound in Theorem \ref{th-Upper-D} $(ii)$ is sharp.
\begin{corollary}\label{cor-14}
$t^3(D_8\Box P_n)=13$ for $n\geq 11$.
\end{corollary}
\begin{proof}
From Proposition \ref{pro-DP-2} and Corollary \ref{th-cor}, we have $t^3(D_8\Box P_n)= c^3(D_8\Box P_n)-1= 14-1=13$.
\end{proof}

\begin{remark}
From Corollary \ref{cor-14}, we have $
t^{3}(D_8\Box P_n)=13=14-1=4\times 5-2\times 1-4-1=2(t^{0}(P_n)+1)(t^{3}(D_8)+1)-(t^0(P_n)+1)\kappa^3(D_8)-(0+1)(3+1)-1$. Note that $p=0$, $q=3$ and $g=3$. This shows that the bound in Theorem \ref{th-Upper-D} $(i)$ is sharp.
\end{remark}

\section{Applications}\label{Section6}

In this section, we demonstrate the usefulness of the proposed constructions by applying them to some instances of Cartesian product networks.

\subsection{2-dimensional grid networks}

A two-dimensional grid graph is the Cartesian product $P_n \square P_m$ of path graphs on $n$ and $m$ vertices. For more details on the grid graph, we refer to \cite{CW98,IR88}.

\begin{proposition}\label{pro-6-1}
For network $P_n \square P_m(n \geq m \geq 5)$,
$$
c^{g}(P_n \square P_m)=
\begin{cases}
3 & \text { if } g=0, \\
5, & \text { if } g=1,\\
8, & \text { if } g=2.
\end{cases}
$$
\end{proposition}
\begin{proof}
Let $H=P_n \square P_m$. If $g=0$, then it follows from Theorem \ref{th-gc-Lower} that
$c^0(H)\geq \kappa^0(H)+g+1=3.$
Let $X=\{(u_1,v_2),(u_2,v_1)\}$; see Fig. \ref{Grid}. Then $H-X$ is not connected, and there is a isolated vertex in it. So $c^0(H)\leq |X|+1=3.$

Suppose that $g=1$. Let $Y=\{(u_{n-1},v_1),(u_{n-1},v_2),(u_n,v_3)\}$; see Fig. \ref{Grid}. Then $H-Y$ is not connected, and the minimum degree of each component is at least $1$. So $c^1(H)\leq |Y|+2=5.$ For any cutset $Y\subseteq V(H)$ with $|Y|\geq 3$, if the minimum degree of each component of $H-Y$ is at least $1$, then $|Y|+|C_1|\geq 3+2=5$, where $C_1$ is the minimum component of $H-Y$. From the arbitrariness of $Y$, we have $c^1(H)\geq 5$.

Suppose that $g=2$. We choose $Z=\{(u_1,v_3),(u_2,v_3),(u_3,v_1),(u_3,v_2)\};$
see Fig. \ref{Grid}.
Then $H-Z$ is not connected, and each component has at least $4$ vertices of degree at least $2$. So $c^2(H)\leq |Z|+4=8.$ For any cutset $Z\subseteq V(H)$, if the minimum degree of each component of $H-Z$ is at least $2$, then $|Z|\geq 4$, and hence $|Z|+|C_1|\geq 4+4=8$, where $C_1$ is the minimum component of $H-Z$. From the arbitrariness of $Z$, we have $c^2(H)\geq 8$.
\end{proof}

The following corollary is immediate by Theorem \ref{th-Upper-gc-condition} and Proposition \ref{pro-6-1}.
\begin{corollary}\label{cor-6-1}
For network $P_n \square P_m(n \geq m \geq 5)$,
$$
t^{g}(P_n \square P_m)=
\begin{cases}
2 & \text { if } g=0, \\
4, & \text { if } g=1,\\
7, & \text { if } g=2.
\end{cases}
$$
\end{corollary}

\subsection{2-dimensional torus networks}

A \textit{torus} is the Cartesian product of two cycles of size at least three, that is, $C_n\Box C_m$; see Fig. \ref{Torus}. Ku et al. \cite{KWH03} showed that there are $n$ edge-disjoint spanning trees in an $n$-dimensional torus.

\begin{proposition}\label{pro-6-2}
For network $C_n \square C_m \ (n \geq m \geq 8)$, we have

$(i)$
$$
c^{g}(C_n \square C_m)=
\begin{cases}
5, & \text { if } g=0 \\
8, & \text { if } g=1\\
12, & \text { if } g=2.
\end{cases}
$$

$(ii)$ $c^{3}(C_n \square C_m)\leq 4m$.
\end{proposition}
\begin{proof}
Let $G=C_n \square C_m$.
If $g=0$, then it follows from Theorem \ref{th-gc-Lower} that
$c^0(G)\geq \kappa^0(G)+g+1=5.$
Let $X=\{(u_1,v_2),(u_2,v_1),(u_1,v_m),(u_n,v_1)\}$. Then $G-X$ is not connected, and there is a isolated vertex in it. So $c^0(G)\leq |X|+1=4+1=5.$

Suppose that $g=1$. We choose an edge $(u_1,v_1)(u_1,v_2)$.  Let
$$
\begin{aligned}
Y
&=N_{G}((u_1,v_1))\cup N_{G}((u_1,v_2))-\{(u_1,v_1),(u_1,v_2)\}\\
&=\{(u_{2},v_1),(u_{2},v_2),(u_n,v_1),(u_n,v_2),(u_1,v_3),(u_1,v_n)\}.
\end{aligned}
$$
Then $G-Y$ is not connected, and the minimum degree of each component is at least $1$. So $c^1(G)\leq |Y|+2=8.$ For any cutset $Y\subseteq V(G)$, suppose that the minimum degree of each component of $G-Y$ is at least $1$.
Recall that $C_1$ is the minimum component of $G-Y$. If there is a component having two vertices, then they are adjacent and their neighbors belong to $Y$, and hence
$|Y|+|C_1|\geq 6+2=8$. If each component has at least three vertices, then their neighbors belong to $Y$, and hence
$|Y|+|C_1|\geq 7+3=10$.
From the arbitrariness of $Y$, we have $c^1(G)\geq 8$.

Suppose that $g=2$. We choose
$$
\begin{aligned}
Z
&=N_{G}((u_1,v_1))\cup N_{G}((u_1,v_2))\cup N_{G}((u_2,v_1))\cup N_{G}((u_2,v_2))-\\
&\{(u_{1},v_1),(u_{2},v_1),(u_1,v_2),(u_2,v_2)\}.
\end{aligned}
$$
Then $G-Z$ is not connected, and the minimum degree of each component is at least $2$ and each component has at least $4$ vertices. So $c^2(G)\leq |Z|+4=8+4=12.$ For any cutset $Z\subseteq V(G)$, suppose that the minimum degree of each component of $G-Z$ is at least $2$.
Then $|Z|\geq 8$, and hence $|Z|+|C_1|\geq 8+4=12$, where $C_1$ is the minimum component of $G-Z$. From the arbitrariness of $Z$, we have $c^2(G)\geq 12$.
\begin{figure}[!htbp]
\centering
\begin{minipage}{0.45\linewidth}
\vspace{3pt}
\centerline{\includegraphics[width=7.5cm]
{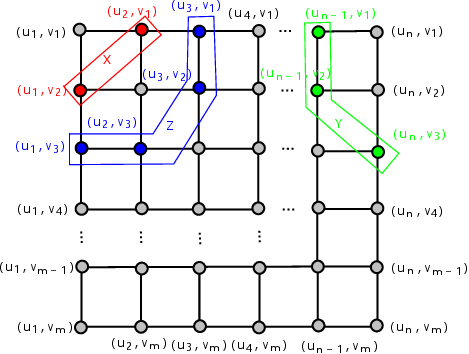}}
\caption{Grid networks.}\label{Grid}
\end{minipage}
\begin{minipage}{0.45\linewidth}
\vspace{3pt}
\centerline{\includegraphics[width=7.5cm]
{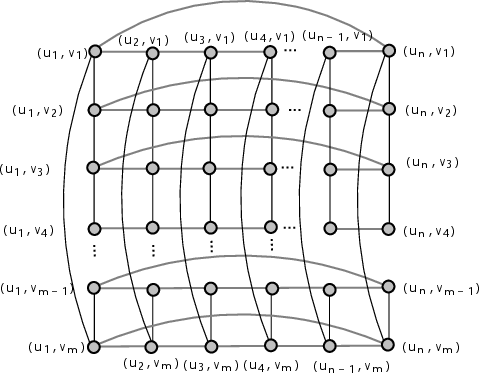}}
\caption{Torus networks.}\label{Torus}
\end{minipage}
\end{figure}

Suppose that $g=3$. We choose
$$
\begin{aligned}
W
&=\{(u_3,v_i)\,|\,1\leq i\leq m\}\cup \{(u_n,v_i)\,|\,1\leq i\leq m\}.
\end{aligned}
$$
Then $G-W$ is not connected, and the minimum degree of each component is at least $3$ and each component has at least $2m$ vertices. So $c^2(G)\leq |W|+2m=2m+2m=4m.$
\end{proof}

\begin{corollary}\label{cor-6-2}
For network $C_n \square C_m \ (n \geq m \geq 8)$, we have

$(i)$
$$
t^{g}(C_n \square C_m)=
\begin{cases}
4, & \text { if } g=0 \\
7, & \text { if } g=1\\
11, & \text { if } g=2.
\end{cases}
$$

$(ii)$ $t^{3}(C_n \square C_m)\leq 4m-1$.
\end{corollary}

\subsection{2-dimensional generalized hypercube}

Let $K_m$ be a clique of $m$ vertices, where $m \geq 2$. An \textit{$n$-dimensional generalized hypercube} \cite{DA97,FA96} is the product of $n$ cliques.

\begin{proposition}\label{pro-6-3}
For network $K_n \square K_m \ (n \geq m \geq 2)$ and $0 \leq g \leq \lfloor (m+n-4)/2\rfloor-1$, we have
$$
\begin{aligned}
n+m-1+g
\leq c^{g}(K_n \square K_m)\leq (g+1)(m-1)+n.
\end{aligned}
$$
Moreover, the bounds are sharp.
\end{proposition}
\begin{proof}
From Theorem \ref{th-gc-Lower}, we have
$c^g(K_n \square K_m)\geq \kappa^g(K_n \square K_m)+g+1
\geq \kappa(K_n \square K_m)+g+1=n+m-1+g.$
Let $G=K_n$ and $H=K_m$. Choose a clique $K_{g+1}$ in $K_{n}$ with $U=\{u_1,u_2,...,u_{g+1}\}$. Let $v_1\in V(H)$ and $N_{H}(v_1)=\{v_2,v_3,...,v_{m}\}$. Let
$$
X=(G(v_1)-K_{g+1}(v_1))\cup \{(u_i,v_j)\,|\,1\leq i\leq g+1,~2\leq j\leq m\}.
$$
Then $X$ is a $g$-good neighbor cutset of $G\Box H$, and $K_{g+1}$ is a component of $G\Box H-X$, and hence $c^g(G\Box H)\leq |X|+|K_{g+1}|=(g+1)(m-1)+n$.

For $g=0$, we have $\kappa^g(K_n \square K_m)=n+m-2$, and hence the bounds are sharp.
\end{proof}

The following corollary is immediate by Theorems \ref{th-Upper-gc}, \ref{th-gc-Lower} and Proposition \ref{pro-6-3}.
\begin{corollary}\label{cor-6-3}
For network $K_n \square K_m(n \geq m \geq 4)$ and $0 \leq g \leq \lfloor n/2\rfloor-1$, we have
$$
\begin{aligned}
n+m-2+g
\leq t^{g}(K_n \square K_m)\leq (g+1)(m-1)+n-1.
\end{aligned}
$$
Moreover, the bounds are sharp.
\end{corollary}

\section{Concluding Remark}

In this paper, we focus our attention on the $g$-good-neighbor conditional diagnosability of Cartesian product networks. It is interesting to study the three other main products, including lexicographic, strong and direct products.

\section{Acknowledgments}

This work was supported by National Science Foundation of China
(Nos. 12471329 and 12061059).

%%%%%%%%%%%%%%%%%%%%%%%%%%%%%%%%%%%%%%%%%%%%%%%%%%%%%%%%%%%%%%%%%%%%%%%%%%%%%%%%%%%%%%%%%%%%%%
%%%%%%%%%%%%%%%%%%%%%%%%%%%%%%%%%%%%%%%%%%%%%%%%%%%%%%%%%%%%%%%%%%%%%%%%%%%%%%%%%%%%%%%%%%%%%%
%%%%%%%%%%%%%%%%%%%%%%%%%%%%%%%%%%%%%%%%%%%%%%%%%%%%%%%%%%%%%%%%%%%%%%%%%%%%%%%%%%%%%%%%%%%%%%

\end{document}